\documentclass[aps,pra,twocolumn,notitlepage,nofootinbib]{revtex4-1}
\usepackage{epsfig,amssymb,amsmath,mathrsfs,amsfonts,amsthm,graphicx,float,color,subfigure}

\usepackage[utf8,latin1]{inputenc}               
\usepackage[T1]{fontenc}                           
\usepackage[british]{babel}                      
\usepackage{mathpazo}                          
\usepackage[dvipsnames,x11names]{xcolor}         
\usepackage[babel]{microtype}                            
\usepackage[colorlinks=true,citecolor=green,linkcolor=Purple,urlcolor=Magenta]{hyperref}     

\usepackage{algorithm}
\usepackage{algpseudocode}

 \def\map#1{\mathcal #1}
\def\<{\langle}\def\>{\rangle}

\floatname{algorithm}{Protocol}

\def\Tr{\operatorname{Tr}}\def\:{\hbox{\bf :}}

\def\N{\mathbb N}

\def\grp#1{\mathsf{#1}}

\def\Span{\mathsf{Span}}

\def\spc#1{\mathcal{#1}}

\def\set#1{\mathcal{#1}}

\newtheorem{theo}{{Theorem}}

\newtheorem{remark}{{Remark}}
\newtheorem{lem}{{Lemma}}

\begin{document}
\title{
Representation matching for remote quantum computing}

\author{Yuxiang Yang}
\email{yangyu@ethz.ch}
\affiliation{Institute for Theoretical Physics, ETH Z\"urich}
\affiliation{QICI Quantum Information and Computation Initiative, Department of Computer Science, The University of Hong Kong, Pokfulam Road, Hong Kong}
\author{Masahito Hayashi}
\email{hayashi@sustech.edu.cn}
\affiliation{Shenzhen Institute for Quantum Science and Engineering, Southern University of Science and Technology, Shenzhen, 518055, China}
\affiliation{Guangdong Provincial Key Laboratory of Quantum Science and Engineering,
Southern University of Science and Technology, Shenzhen 518055, China}
\affiliation{Shenzhen Key Laboratory of Quantum Science and Engineering, Southern
University of Science and Technology, Shenzhen 518055, China}
\affiliation{Graduate School of Mathematics, Nagoya University, Nagoya, 464-8602, Japan}

\begin{abstract}
Many quantum computational tasks have inherent symmetries, suggesting a path to enhancing their efficiency and performance.
Exploiting this observation, we propose representation matching, a generic probabilistic protocol for reducing the cost of quantum computation in a quantum network.
We show that the representation matching protocol is capable of reducing the communication or memory cost to almost minimum in various tasks, including remote execution of unitary gate arrays, permutation gates and unitary conjugation, as well as the storage and retrieval of unitary gates.
\end{abstract}

\maketitle

\medskip

\section{Introduction}
The past few years have witnessed tremendous progress in quantum computing and quantum communication. 
The union of technologies from these two directions will lead to a quantum internet \cite{kimble2008quantum}, where remote nodes can exchange quantum data, execute protocols, and share computational power \cite{childs2005secure,broadbent2009universal,fitzsimons2017unconditionally,kashefi2017multiparty} via quantum communication channels (see Figure~\ref{fig-network}).

  \begin{figure}[b]
\subfigure[]{
\includegraphics[width=0.95\linewidth]{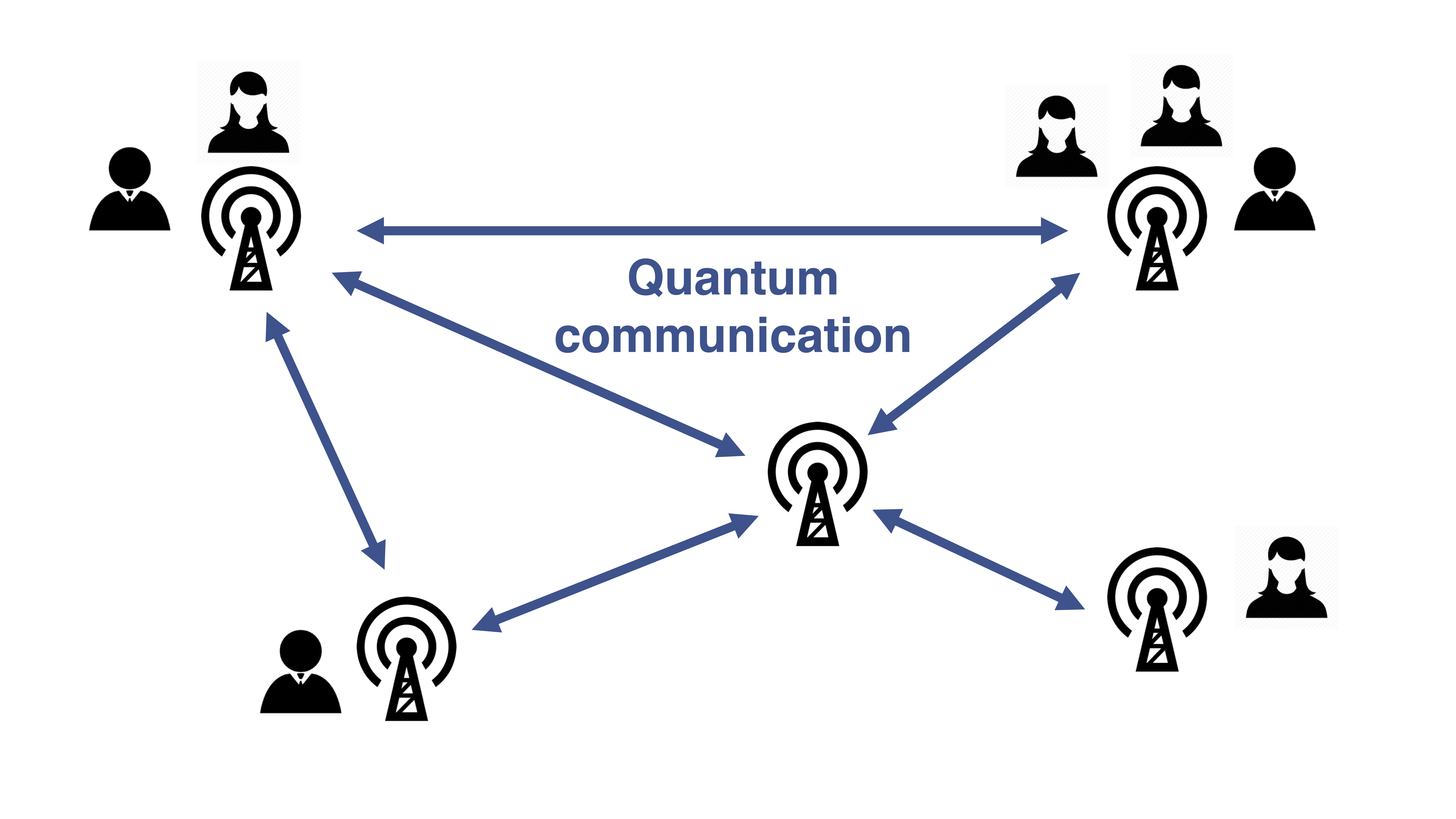}\label{fig-network}}\\
\subfigure[]{
\includegraphics[width=0.95\linewidth]{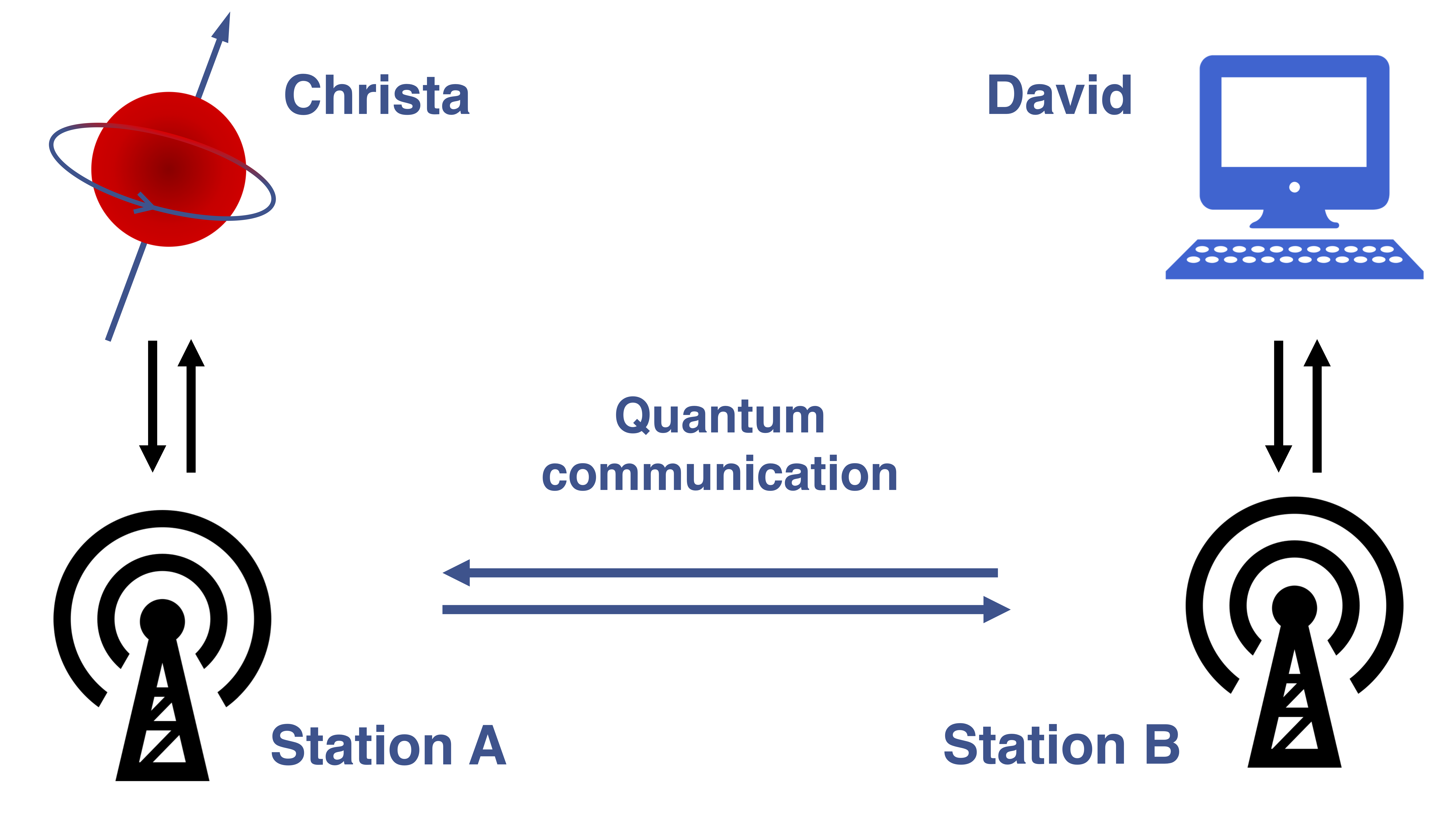}\label{fig-setting}}
\caption{{\bf a) Quantum Internet.} In a quantum network, multiple parties are linked via quantum communication channels. {\bf b) Remote quantum computation.} The task is for David to apply a quantum gate on a quantum state, held by a remote party Christa. The goal is to accomplish this objective while keeping the communication cost as low as possible.}
\end{figure} 

One of the key issues in such a quantum network is the communication cost, quantified by the number of qubits needed to be sent via the communication channels. 
Any proposal of reducing the communication cost will have increasing importance, as the scale of quantum computation and quantum network is expected to increase rapidly in the near future.

To reduce the communication cost, one possibility is to consider probabilistic protocols. 
Many quantum protocols or algorithms, e.g.\ quantum key distribution \cite{brassard1984quantum}, magic state distillation \cite{bravyi2005universal} and error mitigation \cite{temme2017error}, are probabilistic, where one repeats the protocol multiple times until succeed. 
Probabilistic protocols play an pivotal role in circumventing no-go theorems of quantum information \cite{wootters1982single,nielsen1997programmable,buvzek1999optimal}.
 Also, to discuss the classification of quantum complexity classes, e.g., the nondeterministic quantum polynomial-time complexity class  (NQP),
a protocol with a small successful probability plays an important role
\cite[Theorem 7]{QC1}\cite[Theorem 1]{QC2},\cite{QC3}.
Here in the network setting, the remote parties can communicate the desired computation using probabilistic protocols such as gate teleportation \cite{gottesman1999demonstrating,bartlett2003quantum}.
In communication scenarios, probabilistic protocols can overcome limitations including restricted bandwidth and short memory life and accomplish tasks that are impossible for deterministic ones.
Therefore, it may still be beneficial to employ a probabilistic protocol, even if its expected communication cost, i.e. the expected amount of required communication until the protocol succeeds, is higher than deterministic ones. 

Another less explored observation is that many quantum computational tasks are associated with inherent symmetry. In another word, the relevant quantum gates form a group. For example, regarding computations on $n$ qubits, there are the Pauli group $\grp{P}(n)$, the Clifford group $\grp{C}(n)$, the permutation group $\grp{S}(n)$, the braiding group $\grp{B}(n)$, and the special unitary group $\grp{SU}(2)$. 
Many fundamental tasks of quantum information processing adopt group-theoretic structures, including cloning of states \cite{gisin1997optimal,bruss1998optimal,bruss2000phase,chiribella2013quantum} and gates \cite{chiribella2008optimal,dur2015deterministic,chiribella2015universal}, universal gate programming \cite{kim2001storing,vidal2002storing,hillery2002probabilistic,brazier2005probabilistic,ishizaka2008asymptotic,kubicki2019resource,yang2020optimal}, the storage and retrieval of unitary gates \cite{bisio2010optimal,sedlak2019optimal},  and the invertion of general unitary operations \cite{quintino2019reversing}.
A natural question is: Can we utilize the inherent symmetry of a task to enhance the performance and, in particular, to reduce the cost of communication in the remote setting?




In this work, we consider remote quantum computation in a network, where one party, David, would like to run a unitary gate on a quantum state held by Christa, who is far away from him (see Figure \ref{fig-setting}).
We propose representation matching: a generic probabilistic protocol capable of reducing the communication cost of this task, by exploiting the inherent symmetry of the task. We then prove a general lower bound on the communication cost of such setting of remote quantum computing.
When the computation is an array of $n$ identical unitary gates of dimension $d$, our protocol attains the lower bound up to a small overhead that is independent of $n$.
Moreover, the success probability of our protocol is much higher than that of the protocol based on gate teleportation. In particular, the ratio between the two success probabilities grows as $n^{d^2-1}$. Besides, we also apply our protocol to various tasks such as permutational quantum computing  \cite{marzuoli2005computing,jordan2010permutational,planat2017magic,havlivcek2018quantum,ouyang2020faster}, unitary conjugation \cite{chiribella2016optimal,yang2017units,quintino2019reversing,quintino2019probabilistic}, and the storage and retrieval of gates \cite{bisio2010optimal,sedlak2019optimal,yang2020optimal}.

The remaining part of the paper is organised as follows: In Section \ref{sec:prelim} we introduce a few results and notations that are essential for our discussion. In Section \ref{sec-rm} we introduce the setting as well as the representation matching protocol, and in Section \ref{sec-lowerbound} we prove a general lower bound of communication cost under the same setting.  
In Sections \ref{sec-example-compression}, \ref{sec-perm}, \ref{sec-example-conj}, and \ref{sec-sr}, we apply the representation matching protocol to concrete quantum computational tasks. Finally, in Section \ref{sec-conclusion}, we conclude the paper with some discussions on future directions of research.

\section{Preliminary}\label{sec:prelim}

For a Hilbert space $\spc H$ and a vector $|\psi\>  \in  \spc H$, we will use the notation $\psi:  =  |\psi\>\<\psi|$ to denote the projector on the one-dimensional subspace spanned by $|\psi\>$.  
The space of linear operators from a Hilbert space $\spc H$ to another  Hilbert space $\spc K$  will be denoted by $L(\spc H,  \spc K)$.  When the two Hilbert spaces coincide, we will use the shorthand  $L(\spc H): =  L(\spc H,\spc H)$.  In this paper we will focus on finite-dimensional quantum systems,  with $\dim  (\spc H)  <\infty$.  
For a quantum system with Hilbert space $\spc H$, the set of quantum states   will be denoted by  ${\cal S} (\spc H) :  = \{  \rho  \in  L(\spc H) ~|~  \Tr [\rho]=1 \, , \,    \<\psi|  \rho  |\psi\>  \ge 0  \,  \forall |\psi\> \in \spc H \}$.   

 A quantum process that deterministically transforms an input system into a (possibly different) output system is called a {\em quantum channel}.   A quantum channel transforming an  input system with Hilbert space $\spc H^{\rm in}$ into an output system with (possibly different) Hilbert space $\spc H^{\rm out}$ is a completely positive trace-preserving map $\map C:  L( \spc H^{\rm in})  \to L  (\spc H^{\rm out})$.  A probabilistic quantum process transforming an  input system with Hilbert space $\spc H_{\rm in}$ into an output system with (possibly different) Hilbert space $\spc H^{\rm out}$ is called a \emph{quantum operation}, described by a completely positive trace-nonincreasing map $\map M:  L( \spc H^{\rm in})  \to L  (\spc H^{\rm out})$. A quantum operation  takes a quantum state in $\spc{H}^{\rm in}$ as input and produces a sub-normalised state in $\spc{H}^{\rm out}$ as output.


We will frequently consider the Hilbert space $\spc{H}^{\otimes n}$, i.e., the Hilbert space of $n$ identical systems, each with a Hilbert space $\spc{H}$. We will treat it with basic knowledge of representation theory, and we refer interested readers to textbooks, e.g., \cite{fulton-harris} or Chapter 6 of \cite{book-hayashi}, for more information. Here we introduce a few useful results without further explanation.

The structure of $\spc{H}^{\otimes n}$ is characterised by the Schur-Weyl duality, which states that there exists an isometry transforming $\spc{H}^{\otimes n}$ into the block diagonal form:
\begin{align}\label{sw-duality}
\spc{H}^{\otimes n}\simeq\bigoplus_{\lambda\in\set{R}_n}\spc{H}^{\lambda}\otimes \spc{M}^{\lambda},
\end{align}
where $\set{R}_n$ is the collection of all Young diagrams of $n$ boxes, $\spc{H}^{\lambda}$ is the irreducible representation subspace of $\grp{SU(d)}$, the special unitary group of degree $d$, characterised by the Young diagram $\lambda$, and $\spc{M}^\lambda$ is the  irreducible representation subspace of $\grp{S}(n)$, the symmetric group of degree $n$. The isometry, named the Schur transform, can be implemented efficiently on a quantum computer \cite{harrow2005applications,bacon2006efficient,krovi2019efficient}. Since the irreducible representations are in one-to-one correspondence with the Young diagrams, the set $\set{R}_n$ is the collection of all Young diagrams with (at most) $d$ rows and $n$ boxes, defined as
\begin{align}
\set{R}_n:=\big\{\lambda=(\lambda_1,\dots,\lambda_d)~:~&\lambda_i\in\N, \lambda_i\ge\lambda_j,\forall i,j;\nonumber\\
&\qquad\sum_{i=1}^d\lambda_i=n\big\}.\label{def-Rn}
\end{align}

At last, we introduce a few dimensional factors that will be useful. The first one is the total number of irreducible representations in the decomposition (\ref{sw-duality}), i.e., the cardinality of the set $\set{R}_n$. By definition of $\set{R}_n$ [Eq.\ (\ref{def-Rn})], we have the following bound
\begin{align}\label{bound-Rn}
|\set{R}_n|\le (n+1)^{d-1}.
\end{align}
Next, the dimension of a $\grp{SU}(d)$ irreducible representation $\lambda$ can be obtained via the following formula
\begin{align}
d_\lambda=\frac{\prod_{1\le i<j\le d}(\lambda_i-\lambda_j-i+j)}{\prod_{k=1}^{d-1}k!}.
\end{align}
The dual of $d_\lambda$, the dimension $m_\lambda$ of a $\grp{S}(n)$ irreducible representation $\lambda$, has the following expression (see, e.g., \cite{fulton-harris}):
\begin{align}\label{dim-sym-irreps}
m_\lambda=\frac{n!\prod_{1\le j<k\le d}(\lambda_j-\lambda_k-j+k)}{\prod_{i=1}^{d}(\lambda_i+d-i)}.
\end{align}

Denoting by $d_{\rm R}$ the maximum of $d_\lambda$ over $\lambda\in\set{R}_n$, from the above formula we have \cite[Eq.\ (6.16)]{book-hayashi}
\begin{align}\label{dR-bound}
d_{\rm R}:=\max_{\lambda\in\set{R}_n}d_\lambda\le (n+1)^{\frac{d(d-1)}{2}}.
\end{align}
We denote by $d_{\rm tot}$ the sum of all $d_\lambda$ in $\set{R}_n$
\begin{align}\label{dtot-def}
d_{\rm tot}:=\sum_{\lambda\in\set{R}_n}d_\lambda,
\end{align}
 which can be bounded as
\begin{align}\label{dtot-bound}
d_{\rm tot}\le d_{\rm R}|\set{R}_n|\le(n+1)^{\frac{(d+2)(d-1)}{2}}.
\end{align}
Finally, we denote by $d_{\rm tot,sq}$ the sum of the squared dimension of all irreducible representations
\begin{align}\label{dast-def}
d_{\rm tot,sq}:=\sum_{\lambda\in\set{R}_n}d^2_\lambda={n+d^2-1\choose n},
\end{align}
having used \cite[Eq.~(57)]{schur1901klasse}. It follows that
\begin{align}\label{dast-asymp}
\log d_{\rm tot,sq} =(d^2-1)\log n+O(1),
\end{align}
where $O(1)$ denotes a term that does not depend on $n$.

\section{Representation matching protocol}\label{sec-rm}

Consider a common remote quantum computing scenario as in Figure \ref{fig-setting}.
The task is for David to execute a target computation $U^{\rm target}$ on a state $\psi^{\rm in}$, held by another remote party Christa. That is, the final output should be $U^{\rm target}\psi^{\rm in}(U^{\rm target})^\dag$, located at Christa's side.
The goal is to design a protocol for the two stations $A$ and $B$, who provide the data transmission service for Christa and David,  that reduces their total communication cost.
The setting is \emph{blind}, which means that the stations do not know $U^{\rm target}$ or $\psi^{\rm in}$ a priori. This is also the case in most practical applications, since the users of a communication link would usually demand their information to be kept private.

When the target computation is a unitary representation of a group $\grp{G}$ on a fixed Hilbert space, we can express it as
\begin{align}\label{target-decomp}
U_g^{\rm target}=V^\dag\left(\sum_{r\in\set{R}}|r\>\<r|_{\rm I}\otimes (U_g^r)_{\rm R}\otimes (I_{m_r})_{\rm M}\right)V\qquad g\in\grp{G}.
\end{align}
Here $U^r$ is an irreducible representation of $\grp{G}$ indexed by $r$, $V$ is a ($g$-independent) unitary gate, $\{|r\>\}$ is an orthonormal basis for indices of the irreducible representations, and $m_r$ denotes the multiplicity of the irreducible representation $U^r$ in the decomposition of $U$. In Eq.~(\ref{target-decomp}), the first register is referred to as the \emph{index register} {\rm I}, the second the \emph{representation register} {\rm R}, and the third the \emph{multiplicity register} {\rm M}. For instance, for $g\in\grp{SU}(2)$ we have $U_g^{\otimes n}\simeq \sum_{j=0}^{n/2}|j\>\<j|\otimes U^j_g\otimes I_{m_j}$ (assuming for simplicity $n$ to be even), where each irreducible representation $U^j$ is characterised by a spin number $j$. We will discuss more on special unitary groups in later sections.

To fulfil the computational task, a straightforward approach is to communicate both the index register and the representation register, and, as $U_g^{\rm target}$ acts trivially on it, the multiplicity register can be stored locally on Christa's side. The overall transmission cost, in terms of qubits, is thus twice of the cost of transmitting both registers. By merging these two registers into one (see Step 4 of Protocol \ref{protocol-rm} later), the cost can be reduced to
\begin{align}\label{maxcost}
c_{\max}=2\lceil\log d_{\rm tot}\rceil,
\end{align}
where 
\begin{align}\label{dtot-general}
d_{\rm tot}:=\sum_{r\in\set{R}}d_r
\end{align}
 and $\lceil\cdot\rceil$ denotes the ceiling function.\footnote{$\log:=\log_2$.}
In the decomposition (\ref{target-decomp}), the index register {\rm I} contains no information on the desired computation $g$. Meanwhile, it is usually pretty costly to transmit the register $\rm I$. For example, transmission cost of the index register for $U_g^{\otimes n} (g\in\grp{SU}(2), n\gg1)$, i.e., multiple parallel uses of a qubit gate, is nearly half of the total transmission cost. Nevertheless, for any input state living in more than one irreducible representations, the index register is indispensable.

Here we propose a probabilistic protocol for this task. Our protocol reduces the cost by executing the computation on the representation register based on an ansatz of $r$ and postselecting the case where the ansatz holds. In this way, the communication cost can be reduced (see Figure \ref{fig-repmatch}).

\begin{algorithm}[H]
   \begin{algorithmic}[1]
   \Require{  An arbitrary input state $|\psi^{\rm in}\>$ with 
   $V|\psi^{\rm in}\>\in\spc{H}^{\rm I}\otimes\spc{H}^{\rm R}\otimes\spc{R}$ on Christa's side, where $V$ is defined by Eq.~(\ref{target-decomp}). }
   \Ensure{ The quantum state $U_g^{\rm target}|\psi^{\rm in}\>$ at Christa's side with some probability; the successful case is heralded.}
  \State Station $A$ gets an input state, applies $V$, sends only the representation register $\rm R$ and stores locally the index register $\rm I$.
\State Station $B$ prepares the ansatz register $\rm A$ in the state:
\begin{align}\label{ansatz-state}
|f^{\rm ans}\>_{\rm A}:=\frac{1}{\sqrt{|\set{R}|}}\sum_{r\in\set{R}}|r\>_{\rm A}.
\end{align}
\State  Station $B$ performs $V^\dag$ [cf.~Eq.~(\ref{target-decomp})] on $\rm A$, $\rm R$, and an ancillary multiplicity register in a trivial state, asks David to perform $U^{\rm target}_g$ and then performs $V$. 
\State Station $B$ sends the output state back to station $A$. Notice that, although station $B$ has to send out two registers $\rm A$ and $\rm R$, the required communication is not necessarily the size of the two registers. Explicitly, since the state to be sent lives in a subspace of $\spc{H}^{\rm A}\otimes\spc{H}^{\rm R}$, there is an isometry $W: |r\>_{\rm A}\otimes |m\>_{\rm R}\to |l\>_{\rm C}$ encoding both registers into a memory, where $\{|m\>\}_{m=1}^{\max d_r}$ is a basis of the representation register and $\{|l\>\}_{l=0}^{\infty}$ is the Fock basis. The part where the state has no support can be cut off and the remaining part of the memory only has dimension $d_{\rm tot}$.
\State {\bf Coherent matching test.} Station $A$ restores both registers from the memory by applying $W^\dag$ and then performs the coherent matching test, which is a quantum operation $\{\map{M}_{\rm yes}, \map{M}_{\rm no}\}$, jointly on the ansatz register $\rm A$ and the index register $\rm I$, where 
\begin{align}
\map{M}_{\rm yes}(\cdot)&:=M_{\rm yes}(\cdot)M_{\rm yes}^\dag\\ 
M_{\rm yes}&:=\sum_{r}\<r|_{\rm A}\otimes|r\>\<r|_{\rm I}.
\end{align} 
\State Station $A$ applies $V^\dag$ and returns the state to Christa if the measurement outcome is ``yes''; otherwise restart from Step 1.
   \end{algorithmic}
   \caption{Representation matching protocol.}\label{protocol-rm}
\end{algorithm}

\begin{figure}  [bt]
\centering
  \includegraphics[width=\linewidth]{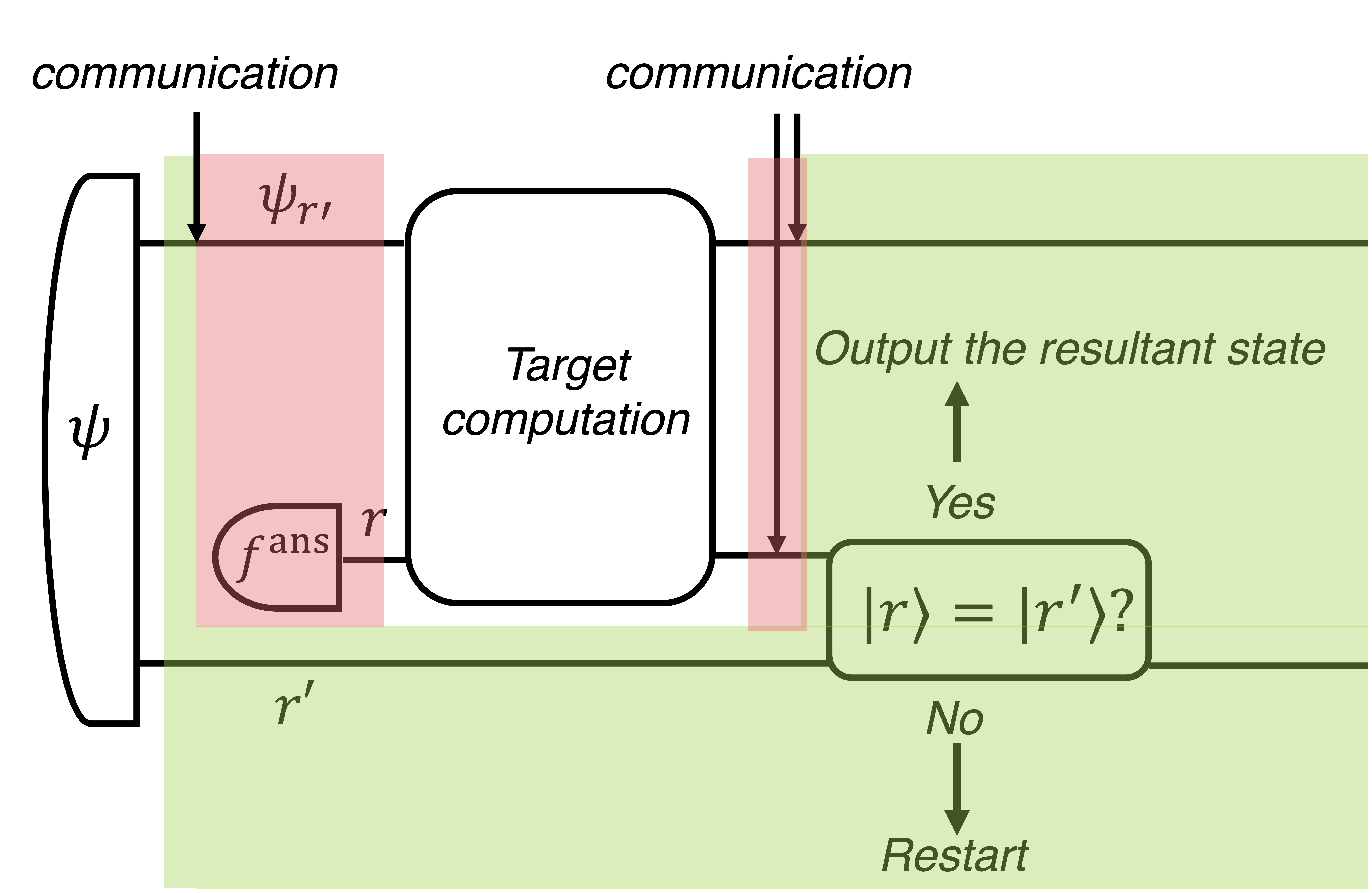}
\caption{{\bf Representation matching protocol.} The above figure illustrates the procedure of the representation matching protocol. The part in green describes the action of station $A$ and the part in red describes the action of Station $B$. Station $A$ sends only the representation register and stores the index register locally. Station $B$ prepares an ansatz state $|f^{\rm ans}\>$ and execute the target computation jointly on it as well as the representation register received from $A$. Station $B$ then returns the output to $A$, who performs a coherent matching test to see if the target computation has been performed correctly. In this way, the cost of transmitting the index register can be waived.  }\label{fig-repmatch}
\end{figure}

In Protocol \ref{protocol-rm}, the communication cost consists of two parts: the cost of sending $\rm R$, which has dimension  
\begin{align}\label{dR}
d_{\rm R}:=\max_{r\in\set{R}}d_r,
\end{align}
to station $B$ and the cost of sending both $\rm A$ and $\rm R$, which have dimension $d_{\rm tot}$ in total, back to station $A$.
Since the index register is only transmitted once,
the communication cost of the representation matching protocol is 
\begin{align}\label{costrm}
c_{\rm rm}=\lceil\log d_{\rm R}\rceil+\lceil\log d_{\rm tot}\rceil,
\end{align}
where $d_{\rm R}$ is defined by Eq.~(\ref{dR}) and $d_{\rm tot}$ is defined by Eq.~(\ref{dtot-general}).
Compared to Eq.~(\ref{maxcost}), the representation matching protocol achieves a reduction of
\begin{align}\label{costsaving}
\Delta c:=c_{\max}-c_{\rm rm}=\lceil\log d_{\rm tot}\rceil-\lceil\log d_{\rm R}\rceil
\end{align}
qubits. The price of the reduction is a risk of failure. The success probability of the representation matching protocol can be straightforwardly evaluated as
\begin{align}\label{probrm}
p_{\rm rm}=\frac{1}{|\set{R}|}.
\end{align}

There maybe some scenarios where the input state $\psi^{\rm in}$ from Christa contains precious information, e.g., outcomes of previous computations. 
As a result, Christa might not know $\psi^{\rm in}$ and cannot reprepare it arbitrarily. 
Since the representation protocol is probabilistic, it seems that $\psi^{\rm in}$ will be lost or corrupted if the protocol fails. 
Given this concern, Christa may want to avoid the corruption of $\psi^{\rm in}$. 
However, in the following we show that, if we modify the coherent matching test slightly, we can avoid corrupting $\psi^{\rm in}$
even when the representation protocol fails. 
The modification is based on the following observation: The stations extract no information on $\psi^{\rm in}$ by doing the coherent matching test, because the success probability (\ref{probrm}) and
the measurement outcome are independent of it 
(even though the post-measurement state does depend on $\psi^{\rm in}$).

To modify the coherent matching test, 
we now specify the quantum operation $\map{M}_{\rm no}$ for the failure case. We assume an additive group structure for $\set{R}$ and define the unitary $V_{\hat{r}}:=\sum_{r \in \set{R}}|r+\hat{r}\>\<r|_{\rm I}$
on $\spc{H}^{\rm I}$.
The set of the measurement outcomes in the modified coherent matching test is given as $\set{R}$.
The measurement operation $\map{M}_{\hat{r}}$ corresponding to the outcome $\hat{r} \in \set{R}$ is the following;
\begin{align}
\map{M}_{\hat{r}}(\cdot)&:=M_{\hat{r}}(\cdot)M_{\hat{r}}^\dag\\ 
M_{\hat{r}}&:=\sum_{r}\<r+\hat{r}|_{\rm A}\otimes|r\>\<r|_{\rm I}.
\end{align} 
The outcome $\hat{r}=0$ corresponds to ``yes'', and other outcomes correspond to ``no''.
The protocol continues even if the outcome is ``no''.
When the outcome $\hat{r}$ is observed, the resultant state is 
$V_{-\hat{r}}VU^{\rm target}_gV^\dag V_{\hat{r}}V|\psi^{\rm in}\>$.
Hence, if we apply the unitary $(V_{-\hat{r}}VU^{\rm target}_gV^\dag V_{\hat{r}}V)^{-1}$,
we can recover the original state $\psi^{\rm in}$, which can be considered as 
a special case of the quantum rewinding lemma \cite[Lemma 8]{watrous2009zero}.

When the outcome $\hat{r}$ is not zero, i.e., when the protocol fails,
the parties can go for another round and apply the representation protocol with the target unitary 
$U^{\rm target}_g(V^\dag V_{-\hat{r}}VU^{\rm target}_gV^\dag V_{\hat{r}}V)^{-1}
=U^{\rm target}_gV^\dag V_{-\hat{r}}V(U^{\rm target}_g)^\dag V^\dag V_{\hat{r}}V$
to the above resultant state, i.e.,
David applies the unitary $U^{\rm target}_gV^\dag V_{-\hat{r}}V(U^{\rm target}_g)^\dag V^\dag V_{\hat{r}}V$
in the second round.
When the outcome of the coherent matching test in the second round is
$0$, the resultant state is the desired state.
Therefore, it is possible to repeat our protocol for multiple rounds until it succeeds. 
The probability that the protocol successes within $n$ rounds is $1-(1-p_{\rm rm})^n$, and thus we have:
\begin{remark}
The success probability of Protocol \ref{protocol-rm} can be amplified to $1-\epsilon$ for arbitrarily small $\epsilon$. This requires executing the protocol for $O\left(\frac{\log(\epsilon)}{\log (1-p_{\rm rm})}\right)$ rounds with target gates $U^{\rm target}_g$, $(U^{\rm target}_g)^\dag$, and $V_{r}$
where $p_{\rm rm}$ is given by Eq. (\ref{probrm}).
\end{remark}

In addition to the blind setting,
Protocol \ref{protocol-rm} works in another setting, the so called \emph{visible} setting, in which
station $B$ has access to the target group element $g\in\grp{G}$ and can execute any arbitrary $g$-dependent computation. 
Remember that station $B$ is only given black-box access to $U^{\rm target}_g$ in the blind setting.
Apparently, the visible setting may require a lower communication cost, as the constraint is more relaxed. 
In the following, we show a lower bound on the communication cost in the visible setting, and 
the lower bound is asymptotically achieved by Protocol \ref{protocol-rm} even in the blind setting.

\section{Lower bound on the communication cost}\label{sec-lowerbound}
Here we prove a lower bound on the communication cost of remote quantum computing, where one party performs a computation $U_g^{\rm target}$ for a state held by another far-away party.
We consider \emph{any arbitrary} deterministic or probabilistic protocol, i.e., we will show a bound on the required communication cost that holds no matter how small the success probability is.
We also consider the visible setting where, compared to the previous setting, station $B$ and David are now considered as one single party. 
Notice that any bound for the visible setting also holds in the blind setting, and thus the lower bound can be used to evaluate the performance of representation matching. 

In the visible setting, the action of station $A$ is still required to be independent of the input state, whereas the action of station $B$ (together with David) can be described by a quantum operation $\map{S}_{g}:L(\spc{H}^{{\rm M, in}})\to L(\spc{H}^{{\rm M, out}})$ acting on a memory system, as illustrated in Figure \ref{fig-visible}.
On the side of station $A$, an encoder is first performed to package (part of) the input state up and a decoder is performed after receiving the state from station $B$.
The encoder is given as an isometric quantum channel $\map{E}:L(\spc{H}^{\rm tot})\to L(\spc{H}^{{\rm M, in}} \otimes\spc{H}^{{\rm M}})$.
The decoder, on the other hand, is given as a quantum operation $\map{D}:L(\spc{H}^{{\rm M, out}}\otimes
\spc{H}^{{\rm M}})\to L(\spc{H}^{{\rm tot}})$. Here $\spc{H}^{\rm M}$ denotes the Hilbert space of a local memory. Notice that $\map{E}$ is assumed to be isometric without loss of generality, because any post-selection or partial-trace can be postponed to $\map{D}$.
The dimension of $\spc{H}^{{\rm M},x}$ ($x={\rm in, out}$) is denoted as $d_{{\rm M},x}$.

\begin{figure}  [bt]
\begin{center}
  \includegraphics[width=\linewidth]{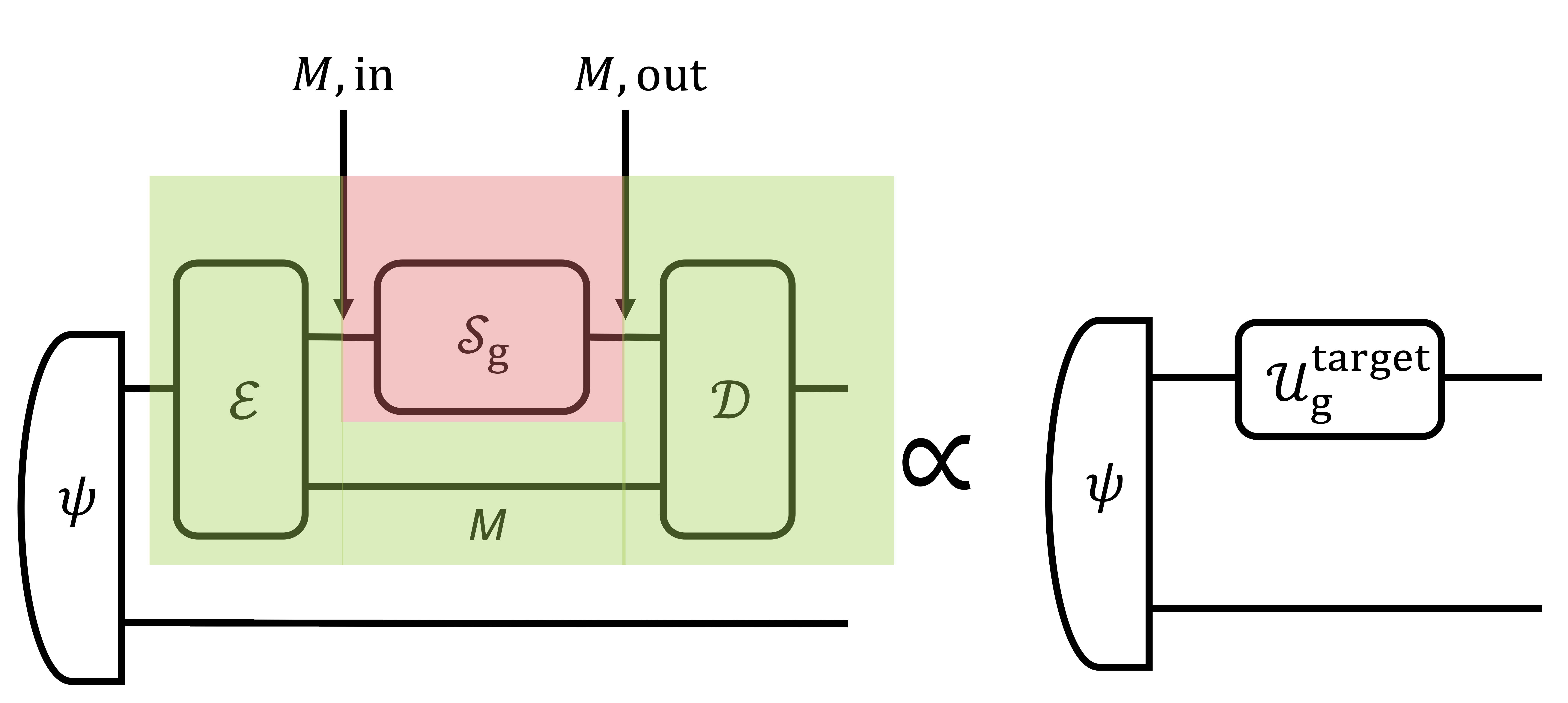}
  \end{center}
\caption{\label{fig-visible}
  {\bf The visible setting of remote quantum computing.} In the visible setting of gate compression, the goal is still to implement $U_{g}^{\rm target}$ on an input $\psi$ held by Christa. The only difference from the blind setting is that station $B$ knows which gate to perform or, equivalently, the value of $g$. Therefore, station $B$ is allowed to perform a generic $g$ dependent quantum operation, which we denote as $\map{S}_{g}$.}
\end{figure}

Here the overall communication cost is $c:=\lceil\log d_{\rm M, in}\rceil+\lceil \log d_{\rm M, out}\rceil$, regarding which we have the following result:
\begin{theo}\label{theo-lowerbound}
Consider probabilistic remote computing of the target gate $U_g^{\rm target}$ as defined by Eq.~(\ref{target-decomp}). The total communication cost in the visible setting is lower bounded as
\begin{align}
c \ge c_{\min}:=\lceil\log d_{\rm tot,sq}\rceil
\end{align}
regardless of the success probability, where 
\begin{align}
d_{\rm tot,sq}:=\sum_{r\in\set{R}}d_r^2.
\end{align}
\end{theo}
The proof can be found in Appendix~\ref{app-lower-bound}. We stress that Theorem \ref{theo-lowerbound} is the ultimate limit for \emph{all protocols}. That is, no matter how small success probability we allow, the lower bound always holds.

By Theorem \ref{theo-lowerbound}, the representation matching protocol (cf.~Protocol \ref{protocol-rm}) has at most an overhead of 
\begin{align}\label{costoverhead}
\delta c:=c_{\rm rm}-c_{\min} = \lceil\log d_{\rm R}\rceil+\lceil\log d_{\rm tot}\rceil-\lceil\log d_{\rm tot,sq}\rceil.
\end{align}
As we will show soon, in many concrete applications the overhead $\delta c$ is only a few qubits, and thus Protocol \ref{protocol-rm} is almost optimal in communication efficiency.




\section{Compression of gate arrays}\label{sec-example-compression}
In the following, we demonstrate several concrete applications of Protocol \ref{protocol-rm}.
The first task we consider is the compression of an array of $n$ identical $d$-dimensional unitary gates, i.e., $U_g^{\otimes n}$ with $g\in\grp{SU(d)}$. By the Schur-Weyl duality [Eq.\ (\ref{sw-duality})], there exists a unitary (the Schur transform) $U_{\rm Sch}$ transforming $U_g^{\otimes n}$ into the block diagonal form:
\begin{align}\label{schur-decomp}
U_{\rm Sch} U_{g}^{\otimes n}U_{\rm Sch}^\dag=\sum_{\lambda\in\set{R}_n}|\lambda\>\<\lambda|_{\rm I}\otimes \left(U_{g}^{\lambda}\right)_{\rm R}\otimes\left(I_{m_\lambda}\right)_{\rm M},
\end{align}
where $U_{g}^{\lambda}$ is now the $\grp{SU(d)}$ irreducible representation characterised by the Young diagram $\lambda$ and $m_\lambda$ is dimension of the $\grp{S}(n)$ representation, now serving as the multiplicity of $U_{g}^{\lambda}$. 

In the remote computing setting, it amounts to that David would like to run $n$ parallel uses of $U_g$ on a remote state of Christa. 
Such a setting is also frequently encountered in quantum sensor networks \cite{komar2014quantum,proctor2018multiparameter}, where the preparation of the sensor state and the application of the unknown unitary gates happen at different locations.
The objective is to communicate $U_g^{\otimes n}$ with lower cost, i.e., to ``compress'' the gate array. 
Protocol \ref{protocol-rm} can be readily applied to fulfil the task, with $U^{\rm target}_g$ [cf.\ Eq.~(\ref{target-decomp})] being $U_g^{\otimes n}$ for $g\in\grp{SU}(d)$.
 
Next, we discuss the performance of Protocol \ref{protocol-rm}.
Using Eq.~(\ref{probrm}), the probability of success is 
\begin{align}
p_{\rm rm}=\frac{1}{|\set{R}_n|}.
\end{align} 
Therefore, by using the bound (\ref{bound-Rn}), we have
\begin{align}\label{p-yes}
p_{\rm rm}\ge\left(\frac{1}{n+1}\right)^{d-1}.
\end{align}

Meanwhile, combining Eq.~(\ref{costrm}) with Eqs.\ (\ref{dR-bound}) and (\ref{dtot-bound}), the communication cost of Protocol \ref{protocol-rm} is upper bounded by
\begin{align}\label{c-tot}
c_{\rm rm}&\le(d^2-1)\log\left(n+1\right)+2,
\end{align}
which, according to Theorem \ref{theo-lowerbound}, attains the optimal scaling in $n$. Let us now examine the communication cost saving $\Delta c$, defined by Eq.~(\ref{costsaving}), which equals the gap between $c_{\rm rm}$ and $c_{\max}$ (the communication cost without compression).  Observing that $d_{\rm tot}\ge d_{\rm tot,sq}/d_{\rm R}$, we have
\begin{align}
\Delta c &\ge \log\left(\frac{d_{\rm tot}}{d_{\rm R}}\right) -1\\
&\ge \log\left(\frac{d_{\rm tot,sq}}{2d^2_{\rm R}}\right) \\
&\ge \log\left(\frac{{n+d^2-1\choose d^2-1}}{2(n+1)^{d(d-1)}}\right),
\end{align}
having used Eqs.~(\ref{dR-bound}) and (\ref{dast-def}) in the last step. In the large $n$ asymptotics, the above implies that
\begin{align}
\Delta c &\ge (d-1)\log n-O(1).
\end{align}
On the other hand, employing Eqs.~(\ref{bound-Rn}) and (\ref{dtot-bound}), we have
\begin{align}
\Delta c &\le \log\left(\frac{d_{\rm tot}}{d_{\rm R}}\right) +1\\
&\le \log\left|\set{R}_n\right|+1 \\
&\le (d-1)\log (n+1)+1.
\end{align}
Summarising the above inequalities, we have identified the scaling of $\Delta c$ as
\begin{align}\label{Deltac-sud}
\Delta c &= (d-1)\log n+O(1).
\end{align}

  \begin{figure}[t]
\subfigure[]{
\includegraphics[width=0.95\linewidth]{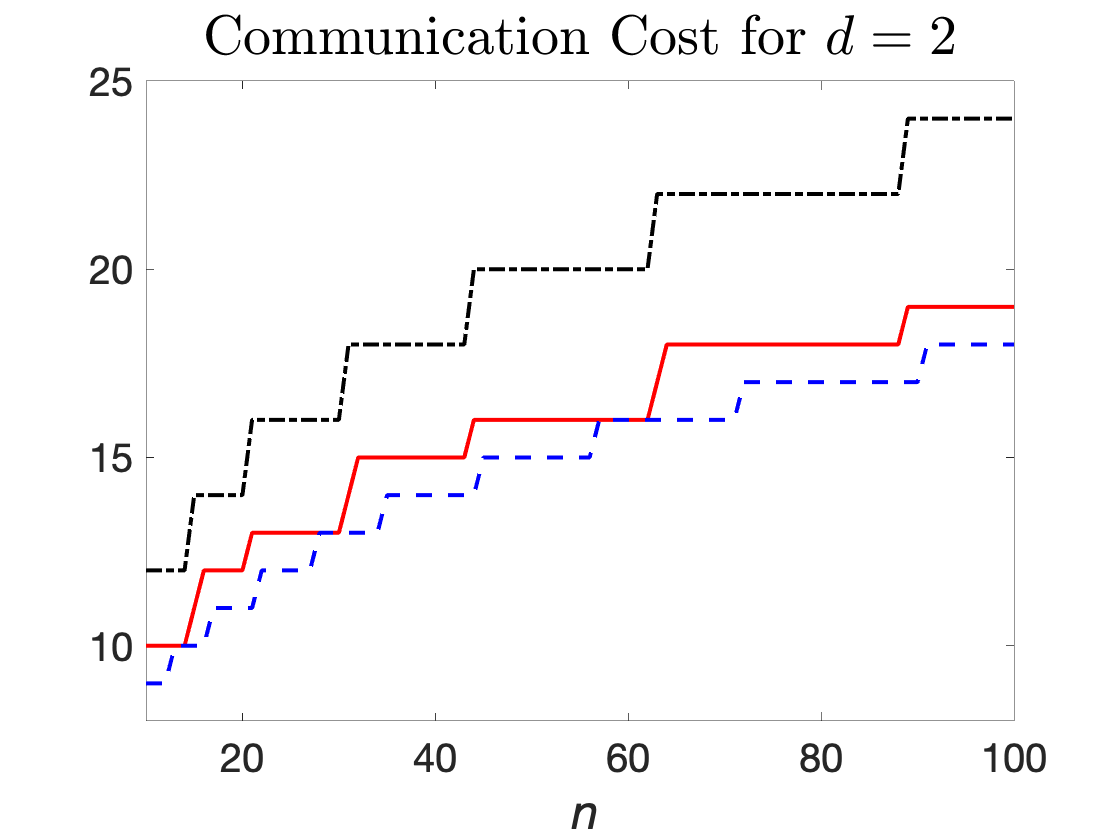}\label{d=2}}\\
\subfigure[]{
\includegraphics[width=0.95\linewidth]{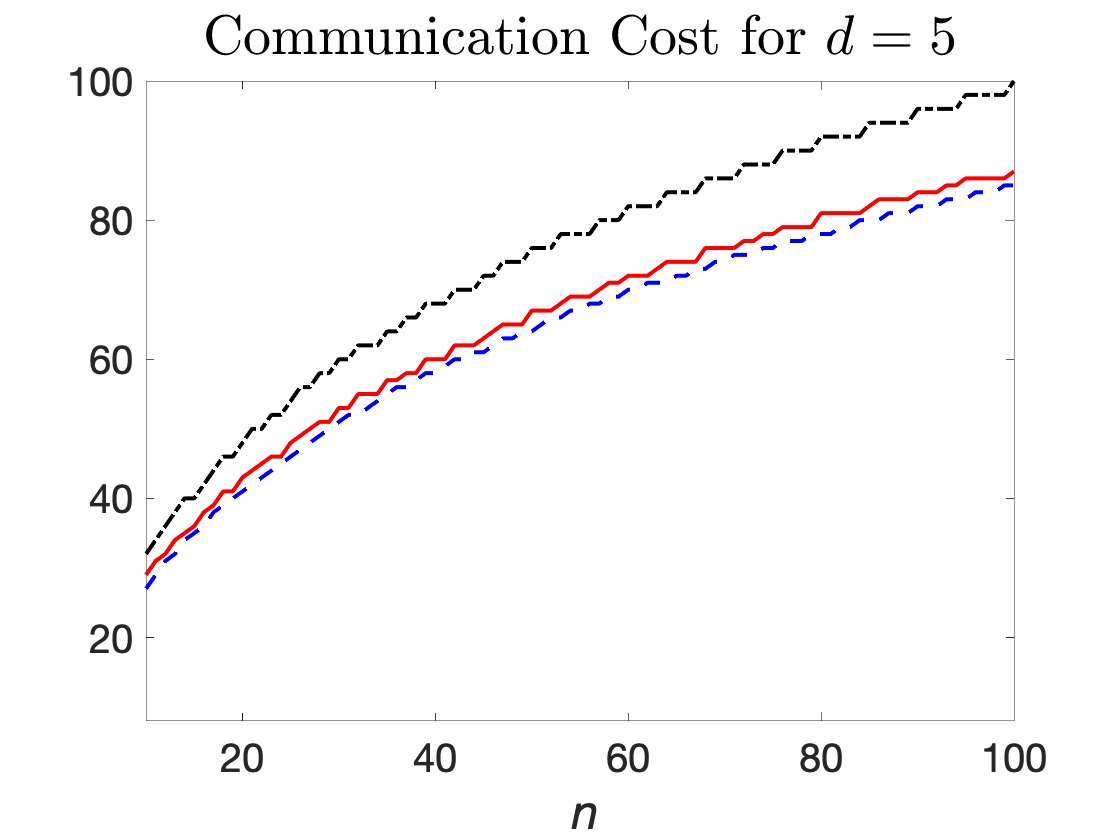}\label{d=5}}
\caption{{\bf Communication cost of Protocol \ref{protocol-rm}.} Communication cost of compressing $U_g^{\otimes n}$ ($g\in\grp{SU}(d)$) is plotted as a function of $n$ for (a) $d=2$ and (b) $d=5$. In both plots, the red, solid lines represent $c_{\rm rm}$, the cost of representation matching (Protocol \ref{protocol-rm}), the blue, dashed lines represent $c_{\min}$, the lower bound of the cost given by Theorem \ref{theo-lowerbound}, and the black, dash-dotted lines represent $c_{\max}$, the original cost. Notice that Protocol \ref{protocol-rm} is asymptotically optimal: at $n=100$, the overhead $\delta c:=c_{\rm rm}-c_{\min}$ is only one for $d=2$ and two for $d=5$.}
\label{fig:costcomparison}
\end{figure} 

Next we show the asymptotic optimality of Protocol \ref{protocol-rm}.
According to Eq.~(\ref{dast-asymp}), the gap between the cost of Protocol \ref{protocol-rm}  and the lower bound $c_{\min}$ is at most 
\begin{align}\label{delta-sud}
\delta c=O(1)
\end{align}
 qubits.
In particular, consider the qubit case, where $g\in\grp{SU(2)}$. Assuming first $n$ to be even, the lower bound can be explicitly evaluated as
\begin{align}
c_{\min} &=\left\lceil \log\left(\sum_{j=0}^{n/2}(2j+1)^2\right)\right\rceil\nonumber\\
& \ge \log\left(\frac{1}{6}(n+1)(n+2)(n+3)\right).
\end{align}
On the other hand, the cost of Protocol \ref{protocol-rm} is
\begin{align}
c_{\rm rm}&=\left\lceil\log\left(n+1\right)\right\rceil+\left\lceil\log\left(\sum_{j=0}^{n/2}(2j+1)\right)\right\rceil\nonumber\\
&\le\log\left(\frac{1}{4}(n+1)(n+2)^2\right)+2.
\end{align}
The overhead is
\begin{align}
\delta c=c_{\rm rm}-c_{\min}\le\log\left(\frac{6(n+2)}{(n+3)}\right)<\log 6.
\end{align}
Similarly, for $n$ odd one can also show that $\delta <\log 6$.
Therefore the overhead is no more than two qubits. In Figure \ref{fig:costcomparison}, we numerically compared the communication cost $c_{\rm rm}$ with $c_{\min}$ and $c_{\max}$, from which one can observe the asymptotic optimality of Protocol \ref{protocol-rm} that matches the above discussion. 


We summarise the performance into the following theorem:
\begin{theo}\label{theo-protocol-performance-compression}
Protocol \ref{protocol-rm} fulfils the task of  compressing $U_g^{\otimes n}$ ($g\in\grp{SU}(d)$) perfectly. The total communication cost is given by Eq.\ (\ref{c-tot}) and attains the optimal scaling in $n$. The success probability is given by Eq.\ (\ref{p-yes}) and scales as $n^{-(d-1)}$.
\end{theo}

We can compare the performance of Protocol \ref{protocol-rm} to another protocol, which is based on the gate teleportation approach \cite{gottesman1999demonstrating,bartlett2003quantum}. One can retrieve a $d$ dimensional unitary gate $U_g$ from the maximally entangled state $|\Phi^+_g\>:=(U_g\otimes I)|\Phi^+\>$ ($|\Phi_+\>:=\sum_i (1/\sqrt{d})|i\>\otimes|i\>$), by performing a generalised Bell test jointly on part of it and an arbitrary input state $|\psi\>$:
\begin{align}
(I\otimes B_j)(|\Phi^+_g\>\otimes|\psi\>)=\frac{1}{d}(U_g\sigma_j|\psi\>)\otimes(I\otimes\sigma_j)|\Phi^+\>
\end{align}
for $j=0,\dots,d^2-1$, where $B_j:=(I\otimes\sigma_j)\Phi^+(I\otimes\sigma_j)$ and $\sigma_j$ is a generalised Pauli operator. In particular, $\sigma_0:=I$ is the identity. Therefore, with probability $(1/d)^2$ we get the outcome $j=0$ and $U_g|\psi\>$ as desired. An naive approach of compressing gate arrays runs as follows:
\begin{enumerate}
  \item Station $B$ let David apply each of the $n$ unitary gates in the array on $|\Phi^+\>$ and sends the resultant state, which is $|\Phi^+_g\>^{\otimes n}$, to the station $A$.
  \item Station $A$ performs gate teleportation locally with the received state and the input from Christa.
\end{enumerate}
Apparently, this approach fares much worse than the representation matching protocol in both the communication cost and the success probability. Indeed, the communication cost is $n\cdot\lceil\log(d^2)\rceil$, which is exponentially higher, and the success probability is $(1/d)^{2n}$, which vanishes exponentially in $n$.

Instead, we can consider an improved version of the gate teleportation based approach by exploiting the decomposition (\ref{schur-decomp}):
\begin{algorithm}[H]
   \begin{algorithmic}[1]
  \State   Station $B$ applies $U_{\rm Sch}$, asks David to perform $U_{g}^{\otimes n}$, and performs $U_{\rm Sch}^\dag$ in sequential order on registers $\rm A$, $\rm R_1$, and $\rm M$ of the state
\begin{align}\label{maxent-state}
|\Phi^+_n\>_{\rm AR_1R_2M}:=\sum_{\lambda\in\set{R}_n}\sqrt{\frac{d_{\lambda}}{d_{\rm tot}}}|\lambda\>_{\rm A}\otimes |\Phi^+_\lambda\>_{\rm R_1R_2}\otimes|\eta_0\>_{\rm M},
\end{align}
where $|\Phi^+_\lambda\>$ is the maximally entangled state on $\spc{H}_\lambda\otimes\spc{H}_\lambda$ and $|\eta_0\>$ is a fixed (but otherwise arbitrary) state. Station $B$ then sends out the resultant state $|\Phi_{{g,n}}\>$ to station $A$.
\State Station $A$ performs a generalised Bell measurement $\{B_j\}_{j=0}^{d_{\rm tot}^2-1}$ ($B_0=|\Phi^+_n\>\<\Phi^+_n|$) jointly on part of $|\Phi_{g,n}\>$ and the input state from Christa.  The protocol is successful if and only if the Bell measurement yields $j=0$.
   \end{algorithmic}
   \caption{Gate teleportation based compression of $U_{g}^{\otimes n}$.}\label{protocol-tele}
\end{algorithm}

  \begin{figure}[t]
\includegraphics[width=0.95\linewidth]{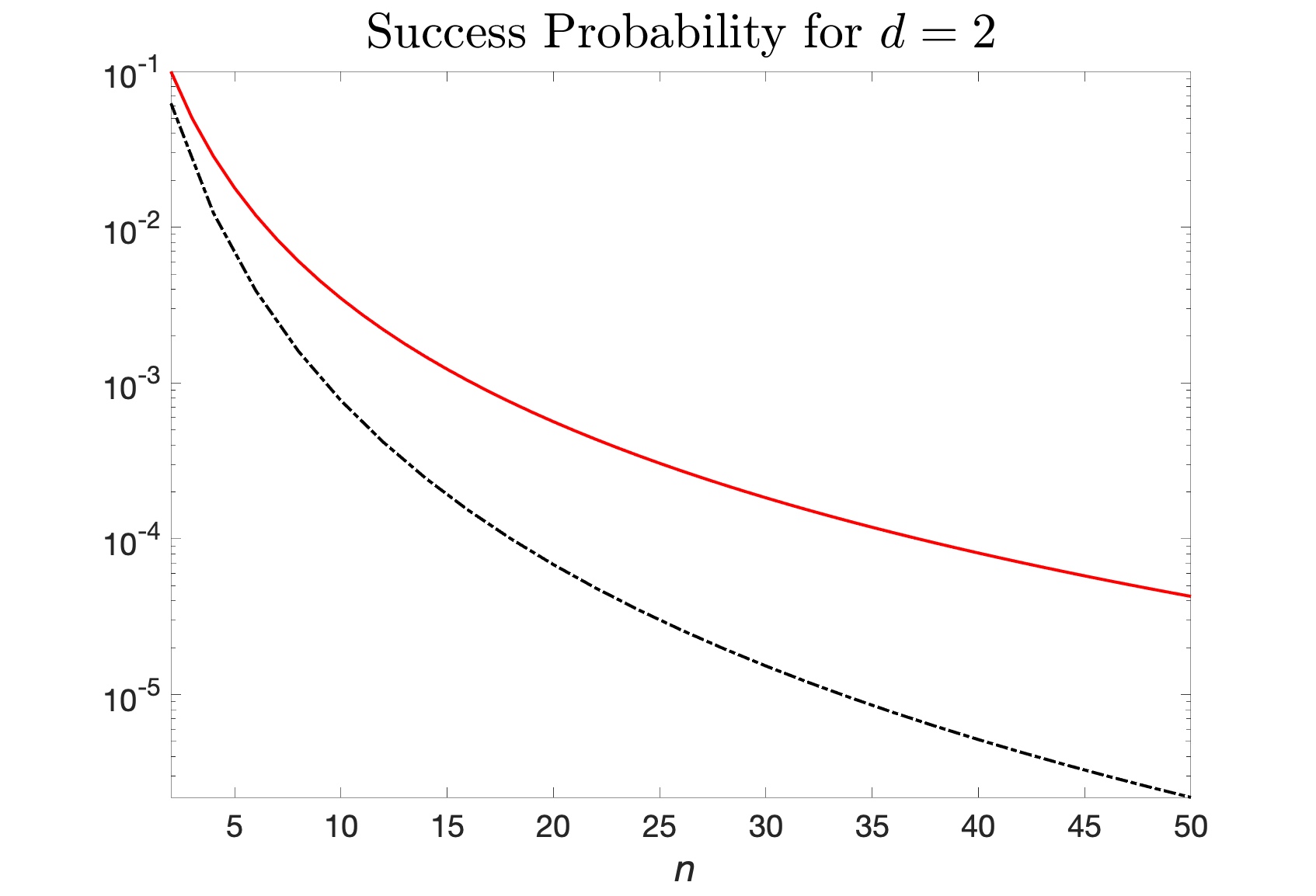}
\caption{{\bf Success probability of Protocols \ref{protocol-rm} and \ref{protocol-tele}.} Success probability of compressing $U_g^{\otimes n}$ ($g\in\grp{SU}(d)$) is plotted as a function of $n$ for $d=2$.  The red, solid line represents $p_{\rm rm}$, the success probability of representation matching (Protocol \ref{protocol-rm}), while the black, dash-dotted line represents $p_{\rm tele}$, the success probability of gate teleportation (Protocol \ref{protocol-tele}).}
\label{fig:probcomparison}
\end{figure}

The communication cost for the teleportation-based protocol is $c_{\min}$, which is very close to that of Protocol \ref{protocol-rm} according to Eq.~(\ref{delta-sud}). However, as the success probability of gate teleportation is inverse proportional to the square of the system dimension, the success probability of Protocol \ref{protocol-tele} is only 
\begin{align}
p_{\rm tele}=\frac{1}{d^2_{\rm tot}} 
\end{align}
where $d_{\rm tot}$ is defined by Eq.\ (\ref{dtot-def}).
In contrast, Protocol \ref{protocol-rm} has a much higher probability of success. The ratio between the two success probabilities is
\begin{align}
\frac{p_{\rm rm}}{p_{\rm tele}}=O\left(n^{d^2-1}\right).
\end{align}
The advantage of Protocol \ref{protocol-rm} is obvious even in the non-asymptotic regime, as illustrated in Figure \ref{fig:probcomparison}.
In conclusion, our representation matching protocol outperforms even the improved version of gate teleportation in the task of gate array compression.

\section{Compression of permutation gates}\label{sec-perm}
Quantum computation consisting of only permutations of subsystems has gained increasing interest for it offers new insight into topological quantum computing \cite{marzuoli2005computing,jordan2010permutational,planat2017magic,havlivcek2018quantum,ouyang2020faster}.
 In the following, we show that Protocol \ref{protocol-rm} does not only apply to $\grp{SU}(d)$, but also to finite groups like the permutation group of $n$ particles $\grp{S}(n)$.

Consider all permutations of $n$ qudits. By the Schur-Weyl duality [cf.\ Eq.~(\ref{sw-duality})], we can decompose a permutation unitary $V_g^n$ with $g\in\grp{S}(n)$ as 
\begin{align}
U_{\rm Sch} V_{g}^{n}U_{\rm Sch}^\dag=\sum_{\lambda\in\set{R}_n}|\lambda\>\<\lambda|_{\rm I}\otimes (I_{d_\lambda})_{\rm M}\otimes (V_{g}^{\lambda,n})_{\rm R},
\end{align}
where $U_{\rm Sch}$ is the Schur transform, $\set{R}_n$ is defined by Eq.~(\ref{def-Rn}), $V_{g}^{\lambda,n}$ is the irreducible representation associated with Young diagram $\lambda$, and $d_\lambda$ is the dimension of the $\grp{SU}(d)$ irreducible representation $U^\lambda$, now serving as the multiplicity of $V_{g}^{\lambda,n}$. In particular, the dimension of $V_{g}^{\lambda,n}$, $m_\lambda$, is given by Eq.~(\ref{dim-sym-irreps}).

Since $m_\lambda$ are usually quite large, even the minimum communication cost, given by Theorem \ref{theo-lowerbound}, grows linearly in $n$ instead of $\log n$. Nevertheless, Protocol \ref{protocol-rm} is still capable of reducing a moderate amount of communication cost. Here we examine the quantity $\Delta c$, defined by Eq.~(\ref{costsaving}), which is the communication cost saving achieved by Protocol \ref{protocol-rm}. In Figure \ref{fig:costperm}, one can see that $\Delta c$ grows as a function of $n$. In the meantime,  the gap $\delta c$ between the cost of Protocol \ref{protocol-rm} and the minimum cost (see Theorem \ref{theo-lowerbound}) remains very low (less than three qubits).

  \begin{figure}[t]
\includegraphics[width=0.95\linewidth]{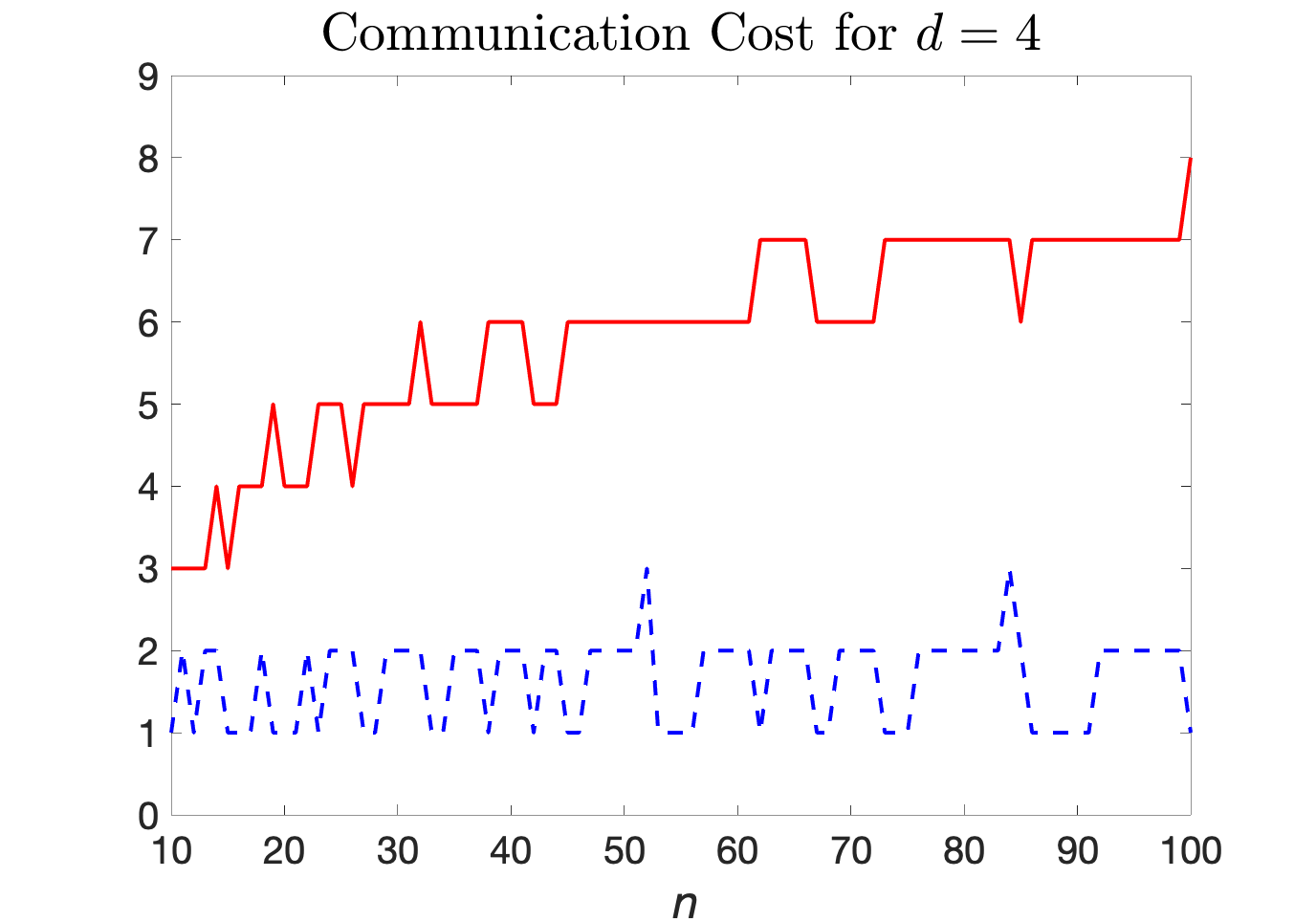}
\caption{{\bf Cost reduction of Protocol \ref{protocol-rm} for permutation gates.} The red, solid line represents the communication cost saving of Protocol \ref{protocol-rm} [cf.\ Eq.~(\ref{costsaving})] for the task of compressing permutation gates of $n$ qudits for $d=4$.  The blue, dashed line represents the gap $\delta c$ [cf.\ Eq.~(\ref{costoverhead})].}
\label{fig:costperm}
\end{figure} 

We can also analytically characterise the scaling of $\Delta c$.
Consider the case of qubits. The dimension of each spin-$j$ irreducible representation (\ref{dim-sym-irreps}) reduces to
\begin{align}
m_j=\frac{2j+1}{n+1}{n+1 \choose \frac{n}{2}-j}.
\end{align}
For large $n$, the maximum of $m_j$ is achieved with $j\sim\sqrt{n}$. On the other hand, one can verify that 
\begin{align}
d_{\rm tot}=\sum_{j=0}^{\frac{n}{2}}m_j=\frac{n+2}{n+1}{n+1 \choose \frac{n}{2}},
\end{align}
having assumed $n$ even for simplicity. By Eq.~(\ref{costsaving}), we have
\begin{align}
\Delta c&\ge \log\left(\frac{d_{\rm tot}}{d_{\rm R}}\right) -1\\
&\ge\frac{1}{2}\log n - O(1).\label{Deltac-sn}
\end{align}
Therefore, the communication cost saving of Protocol \ref{protocol-rm} goes to infinity with $n$. In particular, comparing (\ref{Deltac-sn}) to Eq.~(\ref{Deltac-sud}), we conclude that, for large $n$, the communication cost saving achieved by Protocol \ref{protocol-rm} for permutation gates is at least half of the saving for unitary gate arrays, when $d=2$.

\section{Conjugation of unitary gates}\label{sec-example-conj}

Besides compression, Protocol \ref{protocol-rm} also fits the task of conjugating quantum gates remotely. Consider the same remote quantum computing setting as in previous sections, but the goal now is to execute $n$ parallel uses of the gate $U_g^\ast$ for an unknown $g\in\grp{SU}(d)$ on arbitrary input state of Christa, with David performing $m$ parallel uses of $U_g$. Here $U^\ast$ denotes the complex conjugate of $U$. 


From representation theory, it is know that the complex conjugate of an irreducible representation $\lambda$ of $\grp{SU}(d)$ is isomorphic to the irreducible representation with Young diagram $\bar{\lambda}$. Here $\bar{\lambda}$ is the Young diagram \emph{associated to} $\lambda$, which is obtained by changing the box number of each column from $k$ to $d-k$. That is, for every $\lambda$, there exists a unitary $V^\lambda$ so that
\begin{align}\label{conjugate}
(U^\lambda)^\ast=V^\lambda U^{\bar{\lambda}} (V^\lambda)^\dag.
\end{align}
For example, the two-dimensional irreducible representation of $\grp{SU(2)}$ satisfies $U^\ast=\sigma_y U\sigma_y$, where $\sigma_y$ is the Pauli-$y$ matrix.

In order to apply our approach, we need the computation to include $U^{\bar{\lambda}}$ for every $\lambda\in\grp{SU}(d)$. In particular, $\bar{\lambda}$ associated to $\lambda=(n,0,\dots,0)$ has $(d-1)n$ boxes. It is also straightforward that all other associated Young diagrams have more boxes. Therefore, we need $U_g^{\otimes m}$ with $m=(d-1)n$ to implement the computation $(U_g^\ast)^{\otimes n}$.

The computation can be implemented remotely and probabilistically using Protocol \ref{protocol-rm}, where for Step 3 station $B$   performs (by querying $U_g$ $m$ times from David)
\begin{align}
U&=VU_{\rm Sch} U_g^{\otimes m}U_{\rm Sch}^\dag V^\dag\\
V&:=\sum_{\lambda\in\set{R}_n}|\lambda\>\<\lambda|\otimes V^\lambda
\end{align}
with $V^\lambda$ defined by Eq.~(\ref{conjugate}) and $m=(d-1)n$. The communication cost and the success probability are exactly the same as in the case of unitary gate array compression, and the (asymptotic) optimality can be shown in the same way using Theorem \ref{theo-lowerbound}. Therefore we have:
\begin{theo}\label{theo-protocol-performance-conjugation}
Protocol \ref{protocol-rm} fulfils the task of the gate conversion $U_g^{\otimes m}\to (U_g^\ast)^{\otimes n}$  perfectly for every $m\ge (d-1)n$. The total communication cost is given by Eq.\ (\ref{c-tot}) and attains the optimal scaling in $n$. The success probability is given by Eq.\ (\ref{p-yes}) and scales as $n^{-(d-1)}$.
\end{theo}

\section{Storage and retrieval of gate arrays} \label{sec-sr}

The idea of representation matching applies not only to quantum networks over space, which we have considered in the previous sections, but also to quantum networks over time. Here, as a typical example, we apply representation matching to the task of storage and retrieval of unitary gate arrays \cite{sedlak2019optimal}. As shown in Figure \ref{fig-storageretrieval}, the goal is to store $m$ instances of an unknown unitary gate $U_g$ and to retrieve $n(\le m)$ instances later on. In comparison to the remote quantum computing, the task is for David to apply a computation on a quantum state held by Christa, who locates at the future of David.
This task is also known as quantum learning \cite{bisio2010optimal,mo2019quantum} and has recently been shown to be closely related to the optimal programming of quantum gates \cite{yang2020optimal}.

\begin{figure}  [bt]
\begin{center}
  \includegraphics[width=\linewidth]{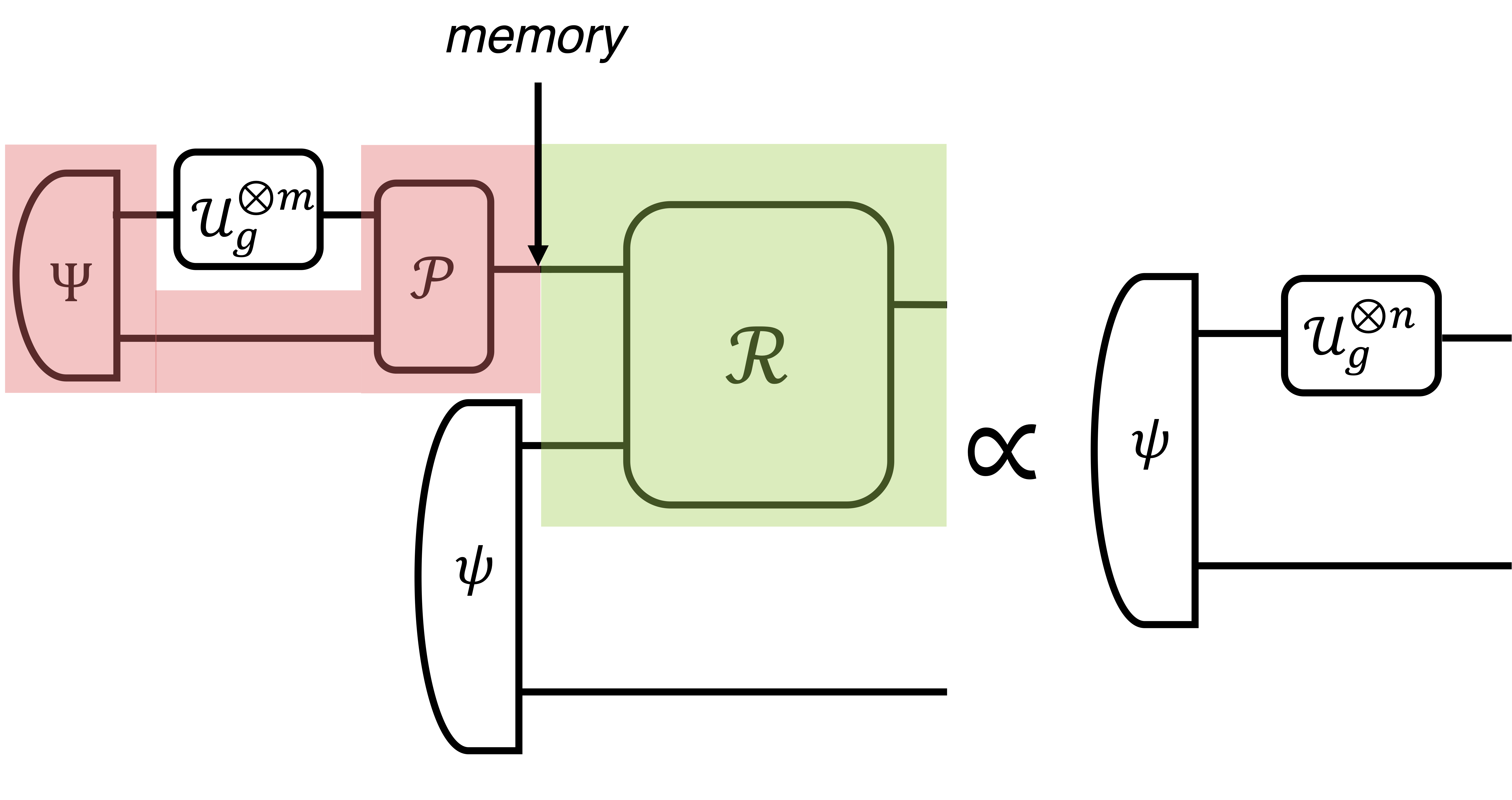}
  \end{center}
\caption{\label{fig-storageretrieval}
  {\bf Storage and retrieval of unitary gates.} The task of storage-and-retrieval of unitary gates is illustrated. The task is separate into two stages: storage (in red) and retrieval (in green). In the storage stage, $U_g^{\otimes m}$ ($g\in\grp{SU}(d)$) is applied on (part of) a state $\Psi$ of the memory, followed by post-processing $\map{P}$. In the retrieval stage, a retrieval operation $\map{R}$ is applied to extract $U_g^{\otimes n}$ from the memory and apply it on an arbitrary input state. The whole process is probabilistic and (heralded) failure is allowed.  }
\end{figure}

Here we employ the idea of representation matching and propose the following protocol for storage and retrieval of gate arrays:
\begin{algorithm}[H]
   \begin{algorithmic}[1] 
\State (Storage.) Apply $(U_{\rm Sch} U_{g}^{\otimes n}U_{\rm Sch}^\dag)_{\rm AR_1M}\otimes I_{\rm R_2}$ [see Eq.~(\ref{schur-decomp})] on the memory state
\begin{align}
|\Psi_n\>_{{\rm AR_1R_2}}:=\sum_{\lambda\in\set{R}_n}\frac{d_{\lambda}}{\sqrt{d_{\rm tot,sq}}}|\lambda\>_{A}\otimes |\Phi^+_\lambda\>_{{\rm R_1R_2}}\otimes|\eta_0\>_{\rm M},
\end{align}
 where $|\eta_0\>$ is an arbitrary fixed state. Store the resultant state  into a memory register.
\State (Retrieval.) To apply the stored unitary on any input state, which can be cast into the form 
\begin{align}\label{psi-in}
|\psi\>=\sum_{\lambda\in\set{R}_n}c_{\lambda}|\lambda\>_{\rm I}\otimes|\psi_\lambda\>_{\rm RP}
\end{align}
where $\rm P$ is a purification register and $\sum_{\lambda\in\set{R}_n}|c_\lambda|^2=1$.
 perform a quantum operation $\{\map{N}_{\rm yes}, \map{N}_{\rm no}\}$ jointly on $\rm A$, $\rm I$, $\rm R$ and $\rm R_2$. The successful operation is defined as 
\begin{align}
\map{N}_{\rm yes}(\cdot)&:=N_{\rm yes}(\cdot)N_{\rm yes}^\dag\\ 
N_{\rm yes}&:=\sum_{\lambda\in\set{R}_n}\<\lambda|_{\rm I_1}\otimes|\lambda\>\<\lambda|_{\rm I}\otimes  \<\Phi^+_{\lambda}|_{\rm RR_2}.
\end{align}
   \end{algorithmic}
   \caption{Probabilistic storage and retrieval of $U_{g}^{\otimes n}$.}\label{protocol-sr}
\end{algorithm}

Next, we analyse the performance of our protocol.
The input state  $|\psi\>$ [see Eq.\ (\ref{psi-in})] and the memory state that stores $U_{g}^{\otimes n}$ can be jointly expressed as
\begin{align} 
|\psi'\>=\sum_{\lambda',\lambda\in\set{R}_n}\frac{c_{\lambda}d_{\lambda'}}{\sqrt{d_{\rm tot,sq}}}|\lambda'\>_{\rm A}\otimes|\lambda\>_{\rm I}\otimes|\psi_\lambda\>_{\rm RP}\otimes|\Phi^+_{g,\lambda'}\>_{\rm R_1R_2}.
\end{align}
The role of the quantum operation $\map{N}_{\rm yes}$ is two-fold: First, it check as if $\rm A$ and $\rm I$ are in the same state; second, if so, it performs the gate teleportation, extracting $U_{g,\lambda}$ from $\rm R_1$ and applying it on $\rm R$. Therefore, Protocol \ref{protocol-sr} integrates our representation matching idea into the gate teleportation approach.
The protocol succeeds, if and only if both the representation matching and the gate teleportation are successfully performed, and then the state becomes 
\begin{align} 
|\psi_{\rm out}\>=\sum_{\lambda\in\set{R}_n}c_{\lambda}|\lambda\>_{\rm I}\otimes(U_{g}^{\lambda}\otimes I_{\rm P})|\psi_\lambda\>_{\rm RP},
\end{align}
which is exactly as desired. Therefore, the protocol has no error.

For the gate retrieval-and-storage protocol,
the probability of success is given by
\begin{align}
p_{\rm rs}&=\Tr\left[\map{N}_{\rm yes}\otimes\map{I}_{\rm P R_1}(\psi)\right]=\frac{1}{d_{\rm tot,sq}}.\label{p-yes-sr}
\end{align}
Meanwhile, the protocol requires a quantum memory of size 
\begin{align}\label{c-rs}
c_{\rm rs}=\left\lceil\log d_{\rm tot,sq}\right\rceil=(d^2-1)\log n+O(1),
\end{align} 
having used Eq.~(\ref{dast-asymp}). 
On the other hand, a lower bound of the required memory size can be determined using Theorem \ref{theo-lowerbound}. Indeed, notice that a storage-and-retrieval protocol can always be employed to fulfil a remote computation. To do so, station $B$ stores the gate in a memory and sends it to station $A$ who retrieve it from the memory later, and the communication cost of such a protocol is equal to the size of the memory. Therefore, the lower bound in Theorem \ref{theo-lowerbound} also applies here to the storage-and-retrieval task. By comparing Eq.~(\ref{c-rs}) and  Theorem \ref{theo-lowerbound} we find that
\begin{align}
c_{\rm rs}=c_{\min}.
\end{align}
Therefore, Protocol \ref{protocol-sr} is optimal in memory efficiency for the task of storage-and-retrieval of unitary gates.
Summarising, we have:
\begin{theo}\label{theo-protocol-performance}
Protocol \ref{protocol-sr} fulfils the task of storage-and-retrieval of $U_g^{\otimes n}$ ($g\in\grp{SU}(d)$) perfectly. The memory cost is given by Eq.\ (\ref{c-rs}) and is optimal. The success probability is given by Eq.\ (\ref{p-yes-sr}) and scales as $n^{-(d^2-1)}$.
\end{theo}

Our result extended that of Ref.~\cite{sedlak2019optimal}, where the $n\to1$ task (i.e., storing $n$ uses and retrieving one use) was considered.
We conclude this section by remarking that adaptations of Protocol \ref{protocol-sr} can be applied in other tasks. 
For instance, one can consider the problem of reversing quantum computation: the gate conversion $U_g^{\otimes n}\to (U_g^\dag)^{\otimes n}$.
To this end, notice that $U^\dag=(U^T)^\ast$, where $U^\ast$ and $U^T$ are the complex conjugate and the transpose, respectively, of $U$. We thus divide the gate reversion task into two separate steps: $U\to U^\ast$ and $U^\ast\to (U^\ast)^T$. The former can be accomplished using techniques in Section \ref{sec-example-conj}, while the latter can be achieved with a variant of Protocol \ref{protocol-sr}.

\section{Conclusion and discussion}\label{sec-conclusion}
We studied how to reduce the communication or memory cost for certain types of computational tasks in quantum internet. Our main contribution is twofold. For one thing, we derived a lower bound on the cost. For another, we proposed representation matching, a generic probabilistic protocol capable of asymptotically achieving our lower bound in many practical scenarios. In addition, the success probability of representation matching is also much higher than protocols based on existing ideas, e.g., gate teleportation.
Compared to existing protocols on remote quantum computing, e.g. \cite{PhysRevA.81.062315,yu2016implementation}, our protocols make use of more specific structure of the problem (i.e. to implement multiple uses of a gate drawn from a group) to achieve higher cost reduction.

From the previous examples, one can see that, while the communication cost reduction grows with $n$, the success probability also vanishes. One may ask if there are certain values of $n$ where even the \emph{average} cost of representation matching, i.e., $c_{\rm rm}/p_{\rm rm}$ can be lower than the deterministic protocol. The answer is negative:
Indeed, according to Eqs.~(\ref{maxcost}), (\ref{costrm}), and (\ref{probrm}), the average cost of representation matching cannot be lower than $c_{\rm max}$ if $|\set{R}|\ge2$. Nevertheless, there are still good reasons, both practically and conceptually, to consider probabilistic protocols like the representation matching. For example, representation matching serves as an excellent means of \emph{deterrence}. Imagine, in the same setting as plotted in Figure \ref{fig-network}, that Christa is now an authority who needs to ensure that David is being honest and is performing the desired computation. She would do this by sending a state to David and let him run some quantum gates on it. If David is caught cheating, he will have to pay an extremely heavy fine. For such a task, it is more favourable to use representation matching rather than deterministic protocols: Even if the protocol has a chance of failure, in which case Christa may fail to detect David's dishonesty, David would not risk it when the fine is high enough. Notice that the coherent matching test (see Protocol \ref{protocol-rm}) is performed \emph{after} David returns the output state, so he cannot predict if the protocol will succeed (or if Christa will decide to perform the coherent matching test at all). Representation matching could also be used in case the bandwidth of the communication channel is limited and deterministic protocols are prohibited by this limitation. At last, it is always fundamentally meaningful to explore the ultimate limits of quantum information theory, even at the  cost of a small success probability. Instances of such research include quantum cloning  \cite{chiribella2013quantum}, quantum metrology \cite{gendra2013quantum}, quantum programming \cite{sedlak2019optimal} and, here in this work, cost reduction of remote quantum computing.

In this work we have been focusing on zero-error protocols, while our idea of representation matching can be extended to the approximate setting, which could be an interesting direction of future research. In particular, the communication cost of deterministic and approximate remote execution of unitary gate arrays has been shown to be closely related to quantum metrology \cite{yang2020communication}, which intrigues the question whether similar phenomenon exists in the probabilistic setting.

The key of point of our protocol is the zero-error property.
When no error is allowed at all and the communication cost is huge,
our method is useful.
To increase the success probability under the zero-error condition, 
we may need to increase the amount of quantum communication.
In this work, we have not studied how much quantum communication 
is needed to achieve certain  success probability under the zero-error condition.
It seems that this tradeoff requires a new technical tool.
Therefore, it is an interesting future problem to study  
the tradeoff between the success probability and
the amount of quantum communication under the zero-error condition.

\section*{Acknowledgement}
We thank Tomoyuki Morimae for helpful comments. YY was supported by the Swiss National Science Foundation via the National Center for Competence in Research ``QSIT" as well as via project No.\ 200020\_165843, the ETH Pauli Center for Theoretical Studies, and  the AFOSR via grant No.\ FA9550-19-1-0202.  
MH was supported in part by Guangdong Provincial Key Laboratory (Grant No.\ 2019B121203002).

\bibliography{ref}
\appendix
\begin{widetext}

\section{Lower bound on the communication cost of zero-error remote computation}\label{app-lower-bound}

Given a group $\grp{G}$,
the task under consideration is to perform  $U^{\rm target}_g$, which is a projective unitary representation on $\spc{H}^{\rm tot}$
for any $g \in \grp{G}$, on a remote state.
This representation contains the irreducible representations
$\{U_{g}^r\}_{r \in \set{R}}$, i.e.,\ 
\begin{align}
U^{\rm target}_g=\bigoplus_{r\in\set{R}}U_{g}^r
\end{align}
where the irreducible representation space of $U_{g}^r$ is ${\cal H}^r$ with dimension $d_r$.

We prepare two lemmas before proving the main result. 
\begin{lem}\label{L0}
The dimension of the linear subspace spanned by $\{\langle u_{l'}| U^{\rm target}_g|u_l\rangle\}_{l,l'}$ as a function space over $\grp{G}$ is $\sum_{r \in \set{R}}d_r^2$, where $|u_l\rangle$ is a basis of $\spc{H}^{\rm tot}$.
\end{lem}
\begin{proof}
This is due to the following basic property of irreducible representations: Matrix elements $(U_{g}^r)_{i,j}$ for different irreducible representations $r$ are linearly independent, and therefore the total dimension is $\sum_{r\in\set{R}} d_r^2$.
\end{proof}

\begin{lem}\label{L1}
Assume that
two linear maps $U$ and $V$ from ${\cal H}_1$ to ${\cal H}_2$ satisfy
\begin{align}
U|\psi\rangle= c_{\psi}V|\psi\rangle.
\end{align}
Also Kernel of $U$ is $\{0\}$.
Then, $c_{\psi}$ does not depend on $\psi$.
\end{lem}

\begin{proof}
Assume that $|\psi_1\>,|\psi_2\>$ are linearly independent.
Considering their superposition, we have
\begin{align}
U (|\psi_1\>+|\psi_2\>)=c_{\psi_1+ \psi_2}V (|\psi_1\>+|\psi_2\>)=
c_{\psi_1}V|\psi_1\>+
c_{\psi_2}V|\psi_2\>.
\end{align}
By assumption, $V|\psi_1\>$ and $V|\psi_2\>$ are also independent. Therefore, the above equalities imply $c_{\psi_1}=c_{\psi_2}=c_{\psi_1+\psi_2}$.
\end{proof}

We consider the less stringent \emph{visible} setting, i.e., when station $B$ knows what $U^{\rm target}_g$ is.  
In the visible setting, the action of station $B$ can be described by a quantum operation $\map{S}_{g}:L(\spc{H}^{{\rm M, in}})\to L(\spc{H}^{{\rm M, out}})$ acting on a memory system, as illustrated in Figure \ref{fig-visible}.
As for station $A$, an encoder is first performed to package (part of) the input state up and a decoder is performed after receiving the state from station $B$.
The encoder is given as an isometric quantum channel $\map{E}:L(\spc{H}^{{\rm M, in}})\to L(\spc{H}^{\rm tot} \otimes\spc{H}^{{\rm M}})$.
The decoder, on the other hand, is given as a quantum operation $\map{D}:L(\spc{H}^{{\rm M, out}}\otimes
\spc{H}_{{\rm R}})\to L(\spc{H}^{\rm tot})$. Notice that $\map{E}$ is assumed to be isometric without loss of generality, because any post-selection or partial-trace can be postponed to $\map{D}$.
The dimension of $\spc{H}^{{\rm M},x}$ ($x={\rm in, out}$) is denoted as $d_{{\rm M},x}$.

Now, we impose the perfect recovery condition, which reads
\begin{align}
U^{\rm target}_g \rho (U^{\rm target}_g)^\dagger=
c_{g,\rho}
\map{D} \circ \map{S}_g \circ \map{E}(\rho) 
\label{H1}
\end{align}
 for any $g \in \grp{G}$ and $\rho \in {\cal S}(\spc{H}^{{\rm M, in}})$. Here $c_{g,\rho}$ is the reciprocal of the success probability of the post-selection.

\begin{theo}\label{theo-general-lowerbound}
Under the condition \eqref{H1}, we have
\begin{align}
d_{{\rm M, in}} d_{{\rm M, out}} \ge \sum_{r \in \set{R}}d_r^2
\label{H10}
\end{align}
\end{theo}

\begin{proof}
First, we invoke Stinespring dilations for all the involved channels (operations). 
We choose the environment systems 
$ \spc{H}_{{\rm E}_1}$
and $ \spc{H}_{{\rm E}_2}$ as well as the post-selection systems $\spc{H}_{{\rm S}_1}$ and $\spc{H}_{{\rm S}_2}$.
Station $B$'s operation can be purified as
\begin{align}
\map{S}_g(\cdot)=\Tr_{{\rm E}_1}\left(\<\psi_{{\rm S}_1}|V_g(\cdot)V_g^\dag|\psi_{{\rm S}_1}\>\right)
\end{align}
for an isometry $V_g:\spc{H}^{{\rm M, in}}\to\spc{H}^{{\rm M, out}} \otimes \spc{H}^{{\rm E}_1}\otimes\spc{H}^{{\rm S}_1}$ and a pure state $|\psi_{{\rm S}_1}\>$ on $\spc{H}^{{\rm S}_1}$. Similarly, for the decoder we have
\begin{align}
\map{D}(\cdot)=\Tr_{{\rm E}_2}\left(\<\psi_{{\rm S}_2}|V_{\map{D}}(\cdot)V_{\map{D}}^\dag|\psi_{{\rm S}_2}\>\right)
\end{align}
with $V_{\map{D}}:\spc{H}^{{\rm M, out}}\otimes\spc{H}^{{\rm M}}\to\spc{H}^{\rm tot} \otimes \spc{H}^{{\rm E}_2}\otimes\spc{H}^{{\rm S}_2}$ being an isometry and $|\psi_{{\rm S}_2}\>$ being a pure state on $\spc{H}^{{\rm S}_2}$.
Moreover, by definition, the encoder is of the form
\begin{align}
\map{E}(\cdot)=V_{\map{E}}(\cdot)V_{\map{E}}^\dag
\end{align}
with $V_{\map{E}}$ being an isometry.

With the above dilations, the condition \eqref{H1} can be rewritten as
\begin{align}
U^{\rm target}_g \rho (U^{\rm target}_g)^\dagger=
c_{g,\rho}
\Tr_{{\rm E}_1,{\rm E}_2} \left(\<\psi_{{\rm S}_1},\psi_{{\rm S}_2}|V_{\map{D}}  V_g  V_{\map{E}}\rho \left(V_{\map{E}}  V_g  V_{\map{D}}\right)^\dagger|\psi_{{\rm S}_1},\psi_{{\rm S}_2}\>\right)
\label{H2}
\end{align}
with a coefficient $c_{g,\rho}$ for $g \in \grp{G}$ and $\rho \in {\cal S}(\spc{H}^{{\rm M, in}})$.

Now, we choose 
a basis on 
$\{|e_{k}\rangle\}$ on 
$ \spc{H}^{{\rm E}_1}$ 
and a basis on 
$\{|f_j\rangle\}$ on 
$\spc{H}^{{\rm E}_2}$.
Then, we have
\begin{align}
U^{\rm target}_g \rho (U^{\rm target}_g)^\dagger=
c_{g,\rho}
\sum_{i,j}   \<e_i,f_j,\psi_{{\rm S}_1},\psi_{{\rm S}_2}|V_{\map{D}}  V_g  V_{\map{E}}\rho \left(V_{\map{E}}  V_g  V_{\map{D}}\right)^\dagger|e_i,f_j,\psi_{{\rm S}_1},\psi_{{\rm S}_2}\>.
\label{H3}
\end{align}
Since $\rho=|\psi\>\<\psi|$ is a pure state,
for $i,j$, we have
\begin{align}
U^{\rm target}_g |\psi\>\<\psi| (U^{\rm target}_g)^\dagger=
c_{g,\psi,i,j}\<e_i,f_j,\psi_{{\rm S}_1},\psi_{{\rm S}_2}|V_{\map{D}}  V_g  V_{\map{E}}|\psi\>\<\psi|\left(V_{\map{E}}  V_g  V_{\map{D}}\right)^\dagger|e_i,f_j,\psi_{{\rm S}_1},\psi_{{\rm S}_2}\>.
\label{H4}
\end{align}
with coefficients $\{c_{g,\psi,i,j}\}$.
Define the map $V_{g,i}:=\<e_i,\psi_{{\rm S}_1}|V_{g}$
from $\spc{H}^{{\rm M, in}}$
to  $\spc{H}^{{\rm M, out}}$ and the map $V_{\map{D},j}:=\<f_j,\psi_{{\rm S}_2}|V_{\map{D}}$ from 
 $\spc{H}^{{\rm M, out}}\otimes \spc{H}_{{\rm M}}$
to  $\spc{H}^{\rm tot}$, which are both linear. Substituting into Eq.\ (\ref{H4}), we have
\begin{align}
U^{\rm target}_g |\psi\>\<\psi| (U^{\rm target}_g)^\dagger=
c_{g,\psi,i,j}
V_{\map{D},j}  V_{g,i}  V_{\map{E}}|\psi\>\<\psi| (V_{\map{D},j}  V_{g,i}  V_{\map{E}})^\dagger 
\label{H5}
\end{align}
Thus, there exists $\theta_{g,\psi,i,j}$ such that
\begin{align}
U^{\rm target}_g |\psi\rangle
=
\sqrt{c_{g,\psi,i,j}}e^{i \theta_{g,\psi,i,j}}
V_{\map{D},j}  V_{g,i}  V_{\map{E}}|\psi\rangle.
\label{H6}
\end{align} 
Due to Lemma \ref{L1}, 
$\sqrt{c_{g,\psi,i,j}}e^{i \theta_{g,\psi,i,j}}$ is independent of $|\psi\>$, and we rename it as
$\alpha_{g,i,j}$.
Thus, we have
\begin{align}
U^{\rm target}_g 
=
\alpha_{g,i,j} V_{\map{D},j}  V_{g,i}  V_{\map{E}}
\label{H7}
\end{align}
Defining the matrix ${V}_{g,i,j}:=\alpha_{g,i,j} V_{g,i} $
from the space $\spc{H}_{{\rm M, in}}$
to the space $\spc{H}_{{\rm M, out}}$, we have
\begin{align}
U^{\rm target}_g 
=
V_{\map{D},j}  V_{g,i,j}  V_{\map{E}}.
\label{H8}
\end{align}

On one hand, the dimension of the linear space spanned by $\{\langle u_{l'}| V_{g,i,j}|u_l\rangle\}_{l,l'}$ 
is at most $d_{{\rm M, in}} d_{{\rm M, out}}$. On the other hand, since $U^{\rm target}_g$ can be obtained by a linear transform on $V_{g,i,j}$, we know that the dimension of $\Span\{\langle u_{l'}| V_{g,i,j}|u_l\rangle\}_{l,l'}$ is lower bounded by the dimension of $\Span\{\langle u_{l'}| U^{\rm target}_g |u_l\rangle\}_{l,l'}$. Applying Lemma \ref{L0} we get Eq.\ 
\eqref{H10}.
\end{proof}
\end{widetext} 
\end{document}